\documentclass[12pt]{article}
\usepackage{amsfonts}
\usepackage{amsmath}
\usepackage{a4,latexsym, amssymb}

\newtheorem{Theorem}{Theorem}
\newtheorem{Lemma}{Lemma}
\newtheorem{Corollary}{Corollary}

\newtheorem{Example}{Example}
\newtheorem{Remark}{Remark}
\newtheorem{Counter Example}{Counter Example}
\newtheorem{Incorrect Result}{Incorrect Result}
\newtheorem{Algorithm}{Algorithm}

\begin{document}
\title{The minimal polynomial of sequence obtained from componentwise linear transformation of linear recurring sequence\thanks{This research is supported in part by the National Natural Science Foundation of
China under the Grants 60872025 and 10990011. }}

\author{
Zhi-Han Gao\thanks{Z.-H. Gao is with the Chern Institute of
Mathematics, Nankai University, Tianjin 300071, P.R. China. E-mail: gaulwy@mail.nankai.edu.cn}
\ and Fang-Wei Fu\thanks{F.-W. Fu is with the Chern Institute of
Mathematics and the Key Laboratory of Pure Mathematics and
Combinatorics, Nankai University, Tianjin 300071, P.R. China. Email:
fwfu@nankai.edu.cn}}

\date{}
\maketitle

\begin{abstract}
Let $S=(s_1,s_2,\ldots,s_m,\ldots)$ be a linear recurring sequence
with terms in $GF(q^n)$ and $T$ be a linear transformation of
$GF(q^n)$ over $GF(q)$. Denote
$T(S)=(T(s_1),T(s_2),\ldots,T(s_m),\ldots)$. In this paper, we first
present counter examples to show the main result in [A.M. Youssef
and G. Gong, On linear complexity of sequences over $GF(2^n)$,
Theoretical Computer Science, 352(2006), 288-292] is not correct
in general since Lemma 3 in that paper is incorrect. Then, we determine the
minimal polynomial of $T(S)$ if the canonical factorization of the
minimal polynomial of $S$ without multiple roots is known and thus
present the solution to the problem which was mainly considered in
the above paper but incorrectly solved. Additionally, as a special
case, we determine the minimal polynomial of $T(S)$ if the minimal
polynomial of $S$ is primitive. Finally, we give an upper bound on
the linear complexity of $T(S)$ when $T$ exhausts all possible
linear transformations of $GF(q^n)$ over $GF(q)$. This bound is
tight in some cases.
\end{abstract}


\baselineskip=18pt

\section{Introduction}

A sequence $S=(s_1,s_2,\ldots,s_m,\ldots)$ with terms in a finite
field $GF(q)$ with $q$ elements is called a linear recurring
sequence over $GF(q)$ with characteristic polynomial
\[f(x)=a_0+a_1x+\cdots+a_mx^m\in GF(q)[x]\] if
\[
a_0s_i+a_1s_{i+1}+\cdots+a_ms_{i+m}=0  \mbox{~for any~} i\geq1.
\]The monic characteristic polynomial of $S$ with least degree is called the minimal polynomial of $S$.
The linear complexity of $S$ is defined as the degree of the minimal
polynomial of $S$. The linear complexity of sequences is one of the
important security measures for stream cipher systems (see
\cite{cdr}, \cite{dxs}, \cite{ru1}, \cite{ru2}). For a general
introduction to the theory of linear feedback shift register
sequences, we refer the reader to \cite[Chapter 8]{ln} and the
references therein.

$S=(s_1,s_2,\ldots,s_m,\ldots)$ is a linear recurring sequence over
$GF(2^n)$. It is well known that $GF(2^n)$ could be viewed as an
$n$-dimensional vector space over $GF(2)$. Let $\{\xi_1,\xi_2,
\ldots, \xi_n\}$ and $\{\eta_1,\eta_2,\ldots,\eta_n\}$ be two bases
of $GF(2^n)$ over $GF(2)$. For the basis $\{\xi_1,\xi_2, \ldots,
\xi_n\}$, each term $s_i$ could be written as
$s_i=s_{i1}\xi_1+s_{i2}\xi_2+\cdots+s_{in}\xi_n$ where $s_{ij}\in
GF(2)$. Let $S'=(s_1',s_2',\ldots,s_m',\ldots)$ where
$s_i'=s_{i1}\eta_1+s_{i2}\eta_2+\cdots+s_{in}\eta_n$. Youssef and
Gong \cite{yg} studied the minimal polynomial of $S'$ if the minimal
polynomial of $S$ without multiple roots is known. Let $T$ be a
linear transformation of $GF(2^n)$ over $GF(2)$ and denote
$T(S)=(T(s_1),T(s_2),\ldots,T(s_m),\ldots)$. It is known that the
effect of changing basis of $GF(2^n)$ over $GF(2)$ is equivalent to
the influence of applying an invertible linear transformation $T$ to
the sequence $S$, that is, there exists an invertible linear
transformation $T$ of $GF(2^n)$ over $GF(2)$ such that $T(S)=S'$. In
this paper, we point out that the main result in \cite{yg} is not
correct in general and consider the corresponding problem discussed
in \cite{yg} for general finite field $GF(q^n)$ and general linear
transformation $T$ which is no longer required to be invertible. In
this case, $S$ is a linear recurring sequence over $GF(q^n)$ and $T$
is a linear transformation of $GF(q^n)$ over $GF(q)$. We determine
the minimal polynomial of $T(S)$ if the canonical factorization of
the minimal polynomial of $S$ without multiple roots is known.
Therefore, we give the correct solution to the problem considered in
\cite{yg}.

Our paper is organized as follows. In Section \ref{pre}, we give
some notations and a number of basic results that will be needed in
this paper. In Section \ref{ce}, we present counter examples to show
that the main result \cite[Theorem 1]{yg} is not correct in general
since \cite[Lemma 3]{yg} is incorrect. In Section \ref{mps}, we
determine the minimal polynomial of $T(S)$ if the canonical
factorization of the minimal polynomial of $S$ without multiple
roots is known and thus present the correct results for the
corresponding problem, which is more general than the
case considered in \cite{yg} but incorrectly solved. In Section
\ref{oms}, we determine the minimal polynomial of $T(S)$ if the minimal
polynomial of $S$ is primitive. In Section
\ref{ublc}, an upper bound on linear complexity of $T(S)$ is given.
This bound is tight in some cases.

\section{Preliminaries}
\label{pre} In this section, we give some definitions needed in this
paper and some lemmas upon which the following sections are
discussed.

Let $GF(q)(x^{-1})$ be the field of formal Laurent series over
$GF(q)$ in $x^{-1}$. This field consists of the elements:
\[
\sum_{i=r}^{\infty}s_ix^{-i}
\] where all $s_i\in GF(q)$ and $r$ is an arbitrary integer. The algebraic operations in $GF(q)(x^{-1})$ are defined in the usual way. The field $GF(q)(x^{-1})$ contains the polynomial ring $GF(q)[x]$ as a subring.
The sequence $S=(s_1,s_2,\ldots,s_m,\ldots)$ over $GF(q)$ could be
viewed as an element of $GF(q)(x^{-1})$:
\[
S(x)=\sum_{i=1}^{\infty}s_ix^{-i}
 \] which is called the generating function associated with $S$.
Then, we have the following lemma.
\begin{Lemma}
\label{Lemma1}{\rm \cite[Lemma 2]{ns}}\quad Let $h(x)\in GF(q)[x]$
be an monic polynomial. Then the sequence
$S=(s_1,s_2\ldots,s_m\ldots)$ of elements of $GF(q)$ is a linear
recurring sequence with minimal polynomial $h(x)$ if and only if
\[
S(x)=\frac{g(x)}{h(x)}
\]
with $g(x)\in GF(q)[x], \deg(g)<\deg(h)$, and $\gcd(g,h)=1$.
\end{Lemma}
Note that it is easy to obtain $g(x)$ if the sequence $S$ and its
minimal polynomial $h(x)$ are known. In the following sections, we
will use $g(x)$ to conduct the computation, particularly in our
main algorithm to be discussed in Section \ref{mps}.

Define a mapping $\sigma$ from $GF(q^n)$ to $GF(q^n)$ as follows:
\[\sigma:\alpha\longrightarrow\alpha^q.\]
It is obvious that $\sigma$ is a field automorphism of $GF(q^n)$.
Then, we can extend $\sigma$ to be a ring automorphism of the
polynomial ring $GF(q^n)[x]$ as follows: For
$f(x)=a_0+a_1x+\cdots+a_mx^m\in GF(q^n)[x]$,
\[\sigma:f(x)\longrightarrow \sigma(f(x))\]where
$\sigma(f(x))=\sigma(a_0)+\sigma(a_1)x+\cdots+\sigma(a_m)x^m$.
Similarly, we can extend $\sigma$ again to be a field automorphism
of the formal Laurent series field $GF(q^n)(x^{-1})$ as follows: For
$S(x)=\sum_{i=r}^{\infty}s_ix^{-i}$,
\[\sigma:S(x)\longrightarrow \sigma(S(x))\]where
$\sigma(S(x))=\sum_{i=r}^{\infty}\sigma(s_i)x^{-i}$. Note that for
any $f(x),g(x)\in GF(q^n)[x]$,
\[
\sigma(f(x)g(x))=\sigma(f(x))\sigma(g(x)).
\] and if $g(x)\neq 0$, we have
\[
\sigma(\frac{f(x)}{g(x)})=\frac{\sigma(f(x))}{\sigma(g(x))}.
\]In this paper, we use only one notation $\sigma$ to represent the
above three mappings and $\sigma^{(k)}$ to represent the $k$th usual
composite of $\sigma$. Note that $\sigma^{(0)}$ is the identity
mapping. We are now able to give some lemmas that will be used to
establish our results in this paper:
\begin{Lemma}
\label{Lemma2}{\rm \cite[Lemma 4]{gf}}\quad Let $f(x)\in
GF(q^n)[x]$. Then $\sigma(f(x))$ is irreducible over $GF(q^n)$
if and only if $f(x)$ is irreducible over $GF(q^n)$.
\end{Lemma}
\begin{Lemma}
\label{Lemma3}{\rm \cite[p.81, Exercise 2.36]{ln}}\quad Note that
$GF(q^n)$ is an $n$-dimensional vector space over $GF(q)$. Then, $T$
is a linear transformation of $GF(q^n)$ over $GF(q)$ if and only if
there uniquely exist $c_0, c_1, \ldots, c_{n-1}\in GF(q^n)$ such
that
\[
T(x)=c_0x+c_1x^q+\cdots+c_{n-1}x^{q^{n-1}}, \;\; x \in GF(q^n).
\]
\end{Lemma}
\begin{Lemma}
\label{Lemma4}{\rm \cite[Theorem 8.57]{ln}}\quad  Let $S_1, S_2,
\ldots, S_k$ be linear recurring sequences over $GF(q)$. The minimal
polynomials of $S_1, S_2, \ldots, S_k$ are $h_1(x), h_2(x), \ldots,
h_k(x)$ respectively. If $h_1(x), h_2(x), \ldots, h_k(x)$ are
pairwise relatively prime, then the minimal polynomial of
$\sum_{i=1}^{k}S_i$ is the product of $h_1(x),h_2(x),\ldots,
h_k(x)$.
\end{Lemma}

\begin{Lemma}
\label{Lemma5} {\rm \cite[Lemma 2]{gf}}\quad Let $S$ be a linear
recurring sequence over $GF(q)$ with minimal polynomial
$h(x)=h_1(x)h_2(x)\cdots h_k(x)$ where $h_1(x),h_2(x),\ldots,h_k(x)$
are monic polynomials over $GF(q)$. If $h_1(x),h_2(x),\ldots,h_k(x)$
are pairwise relatively prime, then there uniquely exist sequences
$S_1,S_2,\ldots,S_k$ over $GF(q)$ such that
\[
S=S_1+S_2+\cdots+S_k
\] and the minimal polynomials
of $S_1,S_2,\ldots,S_k$ are $h_1(x),h_2(x),\ldots,h_k(x)$
respectively.
\end{Lemma}

Below we present a specific method to calculate $S_1,S_2,\ldots,S_k$
from $S$ in order to obtain the minimal polynomial of $T(S)$ in this
paper. For this purpose, we use the left shift operator $L$ to
define characteristic polynomials of linear recurring sequences.
This way to define the characteristic polynomials of linear
recurring sequences is discussed in \cite[Chaper 4]{gg}.
Specifically, define $L$ to be a linear mapping of sequences over
$GF(q)$ as follows:
\[ L:\;
(s_1,s_2,s_3\ldots,s_m,\ldots)\longrightarrow
(s_2,s_3,\ldots,s_{m+1}\ldots).
\]
Let $f(x)=a_0+a_1x+\cdots +a_kx^k\in GF(q)[x]$, then $f(L)$ is a
linear mapping of sequences over $GF(q)$ as follows:
\[
S\longrightarrow f(L)(S)=a_0S+a_1L(S)+\cdots+a_kL^{(k)}(S).
\]
From the definition of $f(L)$, we have $f(L)(S)=0$ if and only if
$f(x)$ is a characteristic polynomial of $S$.

\begin{Algorithm}
\label{Algorithm1} Using the same notations as in Lemma
\ref{Lemma5}, the method to obtain $S_1,S_2,\ldots,S_k$ from $S$ is
as follows:

Step 1. Using the Division Algorithm of polynomials over $GF(q)$, we
obtain the polynomials $u_1(x), u_2(x),\ldots, u_k(x)$ over $GF(q)$
satisfying:
\begin{equation}
\label{gcd}
u_1(x)\prod_{i\neq1}h_i(x)+u_2(x)\prod_{i\neq2}h_i(x)+\cdots+u_k(x)\prod_{i\neq
k}h_i(x)=1.
\end{equation}

Step 2. $S_j=u_j(L)\prod_{i\neq j}h_i(L)(S)$.
\end{Algorithm}
\begin{proof}
Step 1 follows from the fact that $h_1,h_2,\ldots, h_k$ are pairwise
relatively prime. Due to the equation (\ref{gcd}), we have
\[
S=S_1+S_2+\cdots+S_k.
\] Meanwhile, since $h_1(L)h_2(L)\cdots h_k(L)(S)=0$,
we have
\[h_j(L)S_j=h_j(L)u_j(L)\prod_{i\neq j}h_i(L)(S)=0.\]
Thus, $h_j(x)$ is a characteristic polynomial of $S_j$. Assume that
the minimal polynomial of $S_j$ is $h_j^{'}(x)$ which is a monic
divisor of $h_j(x)$ for $1\leq j\leq k$. Then, we have that
$h_1^{'}(x),h_2^{'}(x),\ldots,h_k^{'}(x)$ are pairwise relatively
prime. By Lemma \ref{Lemma4}, the minimal polynomial of $S$
is $\prod_{j=1}^{k}h_j^{'}(x)$. Thus,
\[
h_1^{'}(x)h_2^{'}(x)\cdots h_k^{'}(x)=h_1(x)h_2(x)\cdots h_k(x).
\]Therefore, for $1\leq j\leq k$, we have
\[
h_j^{'}(x)=h_j(x)
\] which completes the proof.
\end{proof}

\section{Counter examples}
\label{ce} In this section, we point out that the main result \cite[Theorem 1]{yg} is not correct in general
since \cite[Lemma 3]{yg} is incorrect. The wrong lemma \cite[Lemma 3]{yg} is as follows:
\begin{Incorrect Result}{\rm \cite[Lemma 3]{yg}}\quad
\label{Incorrect Result1} For $i=1,2,\ldots,k$, let $S_i$ be a
homogeneous linear recurring sequence in $GF(q)$ with minimal
polynomial $h_i(x)$. Then, the minimal polynomial $h(x)$ of the
sequence $S=S_1+S_2+\cdots+S_k, S_i\neq S_j$ for $i\neq j$, is given
by the least common multiple of $h_1,h_2,\ldots, h_k$, $i.e.$,
$h(x)=$\rm{lcm}$[h_1,h_2,\ldots,h_k]$.
\end{Incorrect Result}

Actually, a counter example for this lemma is given in \cite[Example 8.58]{ln}. For the completeness of the paper, we give another similar example:
\begin{Counter Example}
\label{Counter Example1} Let $S_1,S_2$ be two linear recurring
sequences over $GF(2)$ with generating functions
\[S_1(x)=\frac{1}{(x+1)(x^3+x+1)}, \; S_2(x)=\frac{1}{(x+1)(x^2+x+1)}.\]
By Lemma \ref{Lemma1}, we obtain that the minimal polynomials of
$S_1$ and $S_2$ are $(x+1)(x^3+x+1)$ and $(x+1)(x^2+x+1)$, respectively.
Meanwhile, we have
\[
(S_1+S_2)(x)=S_1(x)+S_2(x)=\frac{x^2}{(x^3+x+1)(x^2+x+1)}
\] which implies that the minimal polynomial of $S_1+S_2$ is
$(x^3+x+1)(x^2+x+1)$. However,
$\rm{lcm}$$[(x+1)(x^3+x+1),(x+1)(x^2+x+1)]=(x^3+x+1)(x^2+x+1)(x+1)\neq
(x^3+x+1)(x^2+x+1)$. Therefore, \cite[Lemma 3]{yg} is incorrect.
\end{Counter Example}

Note that the main result \cite[Theorem 1]{yg} is obtained by using
the incorrect \cite[Lemma 3]{yg}. Below we give a counter example to
show that \cite[Theorem 1]{yg} is incorrect in general. In order to
establish our results in this paper, we introduce the following
notations and restate \cite[Theorem 1]{yg} in another way.

For a linear transformation $T$ of $GF(q^n)$ over $GF(q)$, we can
extend $T$ to be a linear transformation of $GF(q^n)(x^{-1})$ over
$GF(q)$ as follows:
\[T: \sum_{i=r}^{\infty}s_ix^{-i}\longrightarrow
\sum_{i=r}^{\infty}T(s_i)x^{-i}.\] Meanwhile, we could also extend
the original $T$ to be a linear transformation of sequences with
terms in $GF(q^n)$ as follows:
\[
T:\quad (s_1,s_2,\ldots,s_m,\ldots)\longrightarrow
(T(s_1),T(s_2),\ldots,T(s_m),\ldots).
\]
In this paper, we use only one notation $T$ to represent the above three
linear transformations. Therefore we have:
\[
T(S(x))=T(S)(x)
\]
for any sequence $S$ over $GF(q^n)$.

Note that the linear transformation $T$ of $GF(q^n)$ over $GF(q)$ is
invertible if and only if $T$ is a linear transformation which
transforms one basis to another. This is the case considered in
\cite{yg}. We are now able to restate \cite[Theorem 1]{yg} in
the following way:
\begin{Incorrect Result}{\rm \cite[Theorem 1]{yg}}\quad
\label{Incorrect Result2} Let $S$ be a linear recurring sequence
over $GF(2^n)$ with minimal polynomial $h(x)$ and $T$ be an
invertible linear transformation of $GF(2^n)$ over $GF(2)$. Assume
by Lemma \ref{Lemma3} that
\[
T(x)=c_0x+c_1x^2+\cdots+c_{n-1}x^{2^{n-1}}\mbox{,~~~}x\in GF(2^n)
\]
where $c_0, c_1, \ldots, c_{n-1} \in GF(2^n)$. Then the minimal polynomial of $T(S)$ is given by
\[
\mbox{\rm{lcm}}[\sigma^{(i_0)}(h(x)),\sigma^{(i_1)}(h(x)),\ldots,\sigma^{(i_r)}(h(x))]
\]where $\{i_0,i_1,\ldots,i_r\}=\{j|c_j\neq 0\}$ and $\sigma$ is
defined in the paragraphs below Lemma \ref{Lemma1}.
\end{Incorrect Result}
\begin{Counter Example}
\label{Counter Example2} Let $T$ be a linear transformation of
$GF(2^4)$ over $GF(2)$ such that:
\[
T(x)=x+x^2+x^{2^2}.
\]Let
$\theta$ be a primitive element in $GF(2^4)$. Let $S$ be a linear
recurring sequence over $GF(2^4)$ given by
\[
S=(\theta^{10},\theta^{5},1,\theta^{10},\theta^{5},1,\theta^{10},\theta^{5},1,\theta^{10},\theta^{5},1,\ldots).
\]
The least period of $S$ is 3 and the minimal polynomial of $S$ is $h(x)=x+\theta^{10}$.
Then, $T$ is invertible and the minimal polynomial of $T(S)$ is not
given by
\begin{equation}\label{pr1}
\mbox{lcm}[h(x),\sigma(h(x)),\sigma^{(2)}(h(x))].
\end{equation}
\end{Counter Example}
\begin{proof}We first show that $T$ is invertible. Note that $T$ is
invertible is equivalent to say that $\mbox{ker}(T)=\{0\}$. Suppose, on
the contrary, that there exists nonzero $b\in GF(2^4)$ such that
$T(b)=b+b^2+b^4=0$. So, $1+b+b^3=0$. Since $x^3+x+1|x^7-1$, then we
have $b^7=1$. Meanwhile, $b^{15}=1$ since $b\in GF(2^4)$. Hence,
$b=1$ since $\mbox{gcd}(7,15)=1$. However, $1$ is not a root of
the polynomial $x+x^2+x^4$. That's a contradiction. So,
$\mbox{ker}(T)=\{0\}$ which implies that $T$ is invertible. Since
$\theta^{15}=1$ and
\[
S=(\theta^{10},\theta^{5},1,\theta^{10},\theta^{5},1,\theta^{10},\theta^{5},1,\theta^{10},\theta^{5},1,\ldots),
\] we have
\[S(x)=\theta^{10}x^{-1}+\theta^{5}x^{-2}+x^{-3}+\theta^{10}x^{-4}+\theta^{5}x^{-5}+\cdots=\frac{\theta^{10}}{x+\theta^{10}}.\] So, the minimal polynomial $h(x)$ of $S$ is $x+\theta^{10}$. Meanwhile, by the
definitions of the extended linear transformation $T$ and the field
automorphism $\sigma$,
$T(S(x))=S(x)+\sigma(S(x))+\sigma^{(2)}(S(x))$. Then,
\begin{eqnarray}
T(S)(x)=T(S(x))&=&\frac{\theta^{10}}{x+\theta^{10}}+\frac{\sigma(\theta^{10})}{\sigma(x+\theta^{10})}+\frac{\sigma^{(2)}(\theta^{10})}{\sigma^{(2)}(x+\theta^{10})}\nonumber\\
&=&\frac{\theta^{10}}{x+\theta^{10}}+\frac{\theta^{5}}{x+\theta^{5}}+\frac{\theta^{10}}{x+\theta^{10}}\nonumber\\
&=&\frac{\theta^5}{x+\theta^{5}}\nonumber.
\end{eqnarray}
So, the minimal polynomial of $T(S)$ is $x+\theta^5$. However,
$$\mbox{lcm}[h(x),\sigma(h(x)),\sigma^{(2)}(h(x))]=\mbox{lcm}[x+\theta^{10},x+\theta^{5},x+\theta^{10}]
=(x+\theta^5)(x+\theta^{10})\neq x+\theta^5.$$
Therefore, \cite[Theorem 1]{yg} is incorrect.
\end{proof}

\section{Minimal polynomials of sequences}
\label{mps} In this section, we consider the more general case in which
the finite field is no longer specified to be $GF(2^n)$ and the
invertibility of the linear transformation $T$ of $GF(q^n)$ over
$GF(q)$ is not required. Recall that the main result \cite[Theorem 1]{yg} is not correct.
We present the correct solution to the corresponding problem studied in
\cite{yg} and determine the minimal polynomial of $T(S)$ if the
canonical factorization of the minimal polynomial of $S$ without
multiple roots is known. By Lemma \ref{Lemma3}, the linear transformation $T$ of $GF(q^n)$ over
$GF(q)$ must be of the form:
\begin{equation}
\label{T}
T(x)=c_0x+c_1x^q+\cdots+c_{n-1}x^{q^{n-1}}\mbox{,~~~}c_i\in GF(q^n).
\end{equation}
Let $k$ be a positive factor of $n$ and $n=kt$.
Then, let $T_{k,j}(x)$ denote the following polynomial:
\begin{equation}
\label{T_{k,j}}
T_{k,j}(x)=c_jx^{q^j}+c_{k+j}x^{q^{k+j}}+\cdots+c_{(t-1)k+j}x^{q^{(t-1)k+j}}\mbox{,~~~$0\leq
j<k$}.
\end{equation}
Throughout the rest of the paper, we will only consider linear
recurring sequences with minimal polynomials that have no multiple
roots, which is the case considered in \cite{yg}.

Recall the definition of $\sigma$ in the paragraphs below Lemma
\ref{Lemma1}. For $f(x)\in GF(q^n)[x]$, we denote $k(f)$ the least
positive integer $l$ such that $\sigma^{(l)}(f(x))=f(x)$. Since
$\sigma^{(n)}(f(x))=f(x)$, we have $k(f)$ always exists.

\begin{Lemma}
\label{Lemma6}{\rm \cite[Lemma 3]{gf}}\quad For any $f(x)\in
GF(q^n)[x]$ and positive integer $l$, $\sigma^{(l)}(f(x))=f(x)$ if
and only if $k(f)|l$.
\end{Lemma}

Note that $k(f)$ divides $n$ for any $f(x)\in GF(q^n)[x]$.

 For any
positive integer $m$, define
\[
(GF(q^n))^{m}=\{(a_1,a_2,\ldots,a_m)\mid a_i\in GF(q^n)
\mbox{~~~for~} 1\leq i \leq m\}.
\]
We define a mapping $\mu$ as follows:
\begin{eqnarray}
\mu: &&\bigcup_{m=1}^{\infty}(GF(q^n))^{m}\longrightarrow \mathbb{Z}\nonumber\\
&&(a_1,a_2,\ldots,a_m)\longrightarrow
 \mu(a_1,a_2,\ldots,a_m)\nonumber
\end{eqnarray}
where \begin{equation}\mu(a_1,a_2,\ldots,a_m)= \left\{
\begin{aligned}
          &0,\quad \mbox{\rm{if}} \  (a_1,a_2,\ldots,a_m)=(0,0,\ldots,0), \\
                  &1,\quad \mbox{\rm{if}} \
                  (a_1,a_2,\ldots,a_m)\neq(0,0,\ldots,0).\nonumber
                          \end{aligned} \right.
\end{equation}
\begin{Theorem}
\label{Theorem1} Let $S$ be a linear recurring sequence over
$GF(q^n)$ with irreducible minimal polynomial $h(x)$. Assume that
$g(x)=b_0+b_1x+\cdots+b_lx^l$ is the polynomial over $GF(q^n)$ such
that $S(x)=g(x)/h(x)$ and $l<\mbox{\rm{deg}}(h(x))$. Let $T$ be a
linear transformation of $GF(q^n)$ over $GF(q)$. Then, the minimal
polynomial of $T(S)$ is given by
\[h(x)^{e_0}(\sigma(h(x)))^{e_1}\cdots(\sigma^{(k(h)-1)}(h(x)))^{e_{k(h)-1}}\]
where
$e_j=\mu(T_{k(h),j}(b_0),T_{k(h),j}(b_1),\ldots,T_{k(h),j}(b_l))$
for $0\leq j<k(h)$, and $T_{k,j}(x)$ is defined by (\ref{T}) and
(\ref{T_{k,j}}).
\end{Theorem}
\begin{proof}
Let $t$ be the positive factor of $n$ such that $tk(h)=n$. Then, by (\ref{T}) and
(\ref{T_{k,j}}), we have
\begin{eqnarray}
T(S)(x)&=&T(S(x))\nonumber\\
       &=&c_0S(x)+c_1\sigma(S(x))+c_2\sigma^{(2)}(S(x))+\cdots+c_{n-1}\sigma^{(n-1)}(S(x))\nonumber\\
       &=&c_0\frac{g}{h}+c_1\frac{\sigma(g)}{\sigma(h)}+c_2\frac{\sigma^{(2)}(g)}{\sigma^{(2)}(h)}+\cdots+c_{n-1}\frac{\sigma^{(n-1)}(g)}{\sigma^{(n-1)}(h)}\nonumber\\
       &\stackrel{(a)}{=}&\sum_{j=0}^{k(h)-1}\frac{c_j\sigma^{(j)}(g)+c_{k(h)+j}\sigma^{(k(h)+j)}(g)+\cdots+c_{(t-1)k(h)+j}\sigma^{((t-1)k(h)+j)}(g)}{\sigma^{(j)}(h)}\nonumber\\
       &\stackrel{(b)}{=}&\sum_{j=0}^{k(h)-1}\frac{T_{k(h),j}(b_0)+T_{k(h),j}(b_1)x+\cdots+T_{k(h),j}(b_l)x^{l}}{\sigma^{(j)}(h)}\nonumber
\end{eqnarray}
where $(a)$ follows from the definition of $k(h)$, $(b)$ follows
from the substitution of $g(x)$. For $0\leq j\leq k(h)-1$, let $S_j$
be the linear recurring sequence over $GF(q^n)$ with the generating function:
\[
\frac{T_{k(h),j}(b_0)+T_{k(h),j}(b_1)x+\cdots+T_{k(h),j}(b_l)x^{l}}{\sigma^{(j)}(h)}.
\] By Lemma \ref{Lemma2}, since
$h$ is monic irreducible, we have that
$h,\sigma^{(1)}(h),\ldots,\sigma^{(k(h)-1)}(h)$ are monic
irreducible polynomials with the same degree. Then, the
minimal polynomial of $S_j$ is $1$ if
$T_{k(h),j}(b_0)+T_{k(h),j}(b_1)x+\cdots+T_{k(h),j}(b_l)x^{l}=0$;
otherwise, the minimal polynomial of $S_j$ is $\sigma^{(j)}(h)$. By
the definition of $k(h)$, we know that
$h,\sigma^{(1)}(h),\ldots,\sigma^{(k(h)-1)}(h)$ are distinct. In
addition, by Lemma \ref{Lemma2}, they are all monic irreducible
polynomials. Thus, they are pairwise relatively prime. Meanwhile,
$T(S)=\sum_{j=0}^{k(h)-1}S_j$. Therefore, by Lemma \ref{Lemma4}, the
minimal polynomial of $T(S)$ is given by
\[
h(x)^{e_0}(\sigma(h(x)))^{e_1}\cdots(\sigma^{(k(h)-1)}(h(x)))^{e_{k(h)-1}}
\]
where
$e_j=\mu(T_{k(h),j}(b_0),T_{k(h),j}(b_1),\ldots,T_{k(h),j}(b_l))$.
This completes the proof.
\end{proof}

Below we use Counter Example \ref{Counter Example2} to illustrate Theorem
\ref{Theorem1}:
\begin{Example}
 \label{Example1}
All notations are the same as in Counter Example \ref{Counter Example2}.
The minimal polynomial of $S$ is $h(x)=x+\theta^{10}$ and \[
S(x)=\frac{\theta^{10}}{x+\theta^{10}}.\] Then, $g(x)=\theta^{10}$
and $k(h(x))=2$. Since $T(x)=x+x^2+x^4$, we have $T_{2,0}(x)=x+x^4$ and
$T_{2,1}(x)=x^2$. Therefore, by Theorem \ref{Theorem1}, the minimal
polynomial of $T(S)$ is given by
\[
(x+\theta^{10})^{e_0}(\sigma(x+\theta^{10}))^{e_1}
\] where $e_0=\mu(T_{2,0}(\theta^{10}))=0$ and $e_1=\mu(T_{2,1}(\theta^{10}))=1$. Thus, the minimal polynomial of $T(S)$ is $x+\theta^5$ which is
the same as the correct result in the proof of Counter Example
\ref{Counter Example2}.
\end{Example}
\begin{Corollary}
\label{Corollary1} Let $S$ be a linear recurring sequence over
$GF(q^n)$ with irreducible minimal polynomial $h(x)$. Then, for any
integer table $\{e_0,e_1,\ldots,e_{k(h)-1}\}$ where $e_j$ is $0$ or
$1$, there exists a linear transformation $T$ of $GF(q^n)$ over
$GF(q)$ such that the minimal polynomial of $T(S)$ is
\[
h(x)^{e_0}(\sigma(h(x)))^{e_1}\cdots(\sigma^{(k(h)-1)}(h(x)))^{e_{k(h)-1}}.
\] Furthermore, the maximal linear complexity of $T(S)$, where $T$ exhausts
all possible linear transformations of $GF(q^n)$ over $GF(q)$, is
$k(h)\mbox{\rm{deg}}(h)$.
\end{Corollary}
\begin{proof}
It is trivial if $S$ is a zero sequence. If $S$ is not a zero
sequence, then there exits nonzero polynomial $g(x)$ over $GF(q^n)$
such that $S(x)=g(x)/h(x)$ where $g(x)=b_0+b_1x+\cdots+b_lx^l$ and
$l<\mbox{deg}(h)$. Since $g(x)\neq 0$, there exists $b_i\neq 0$ for
some $0\leq i\leq l$. Let
\[
T(x)=\sum_{j=0}^{k(h)-1}e_{j}x^{q^j}.
\]
By Theorem \ref{Theorem1}, we have the minimal polynomial of $T(S)$
is:
\[
h(x)^{e_0^{'}}(\sigma(h(x)))^{e_1^{'}}\cdots(\sigma^{(k(h)-1)}(h(x)))^{e_{k(h)-1}^{'}}
\]
where $e_j^{'}=\mu(e_jb_0^{q^j},e_jb_1^{q^j},\ldots,e_jb_l^{q^j})$.
Then, if $e_j=0$, it is obvious that $e_j^{'}=0$; if $e_j=1$, we
have $e_jb_i^{q^j}\neq 0$ which implies $e_j^{'}=1$. Thus,
$e_j^{'}=e_j$. Therefore, the minimal polynomial of $T(S)$ is:
\[
h(x)^{e_0}(\sigma(h(x)))^{e_1}\cdots(\sigma^{(k(h)-1)}(h(x)))^{e_{k(h)-1}}.
\]
In particular, let $e_j=1$ for any $0\leq j\leq k(h)-1$, then there
exists linear transformation $T$ of $GF(q^n)$ over $GF(q)$ such that
the minimal polynomial of $T(S)$ is
\[
h(x)\sigma(h(x))\cdots\sigma^{(k(h)-1)}(h(x))
\]
with degree $k(h)\mbox{deg}(h)$. By Theorem \ref{Theorem1}, the
linear complexity of $T(S)$ is at most $k(h)\mbox{deg}(h)$.
Therefore, such a linear transformation $T$ achieves the maximal
possible linear complexity of $T(S)$.
\end{proof}

By Theorem \ref{Theorem1} and Corollary \ref{Corollary1}, we  have
the following result:
\begin{Corollary}
\label{Corollary2}Let $S$ be a linear recurring sequence over
$GF(q^n)$ with irreducible minimal polynomial $h(x)$. Then, the set
of minimal polynomials of $T(S)$, where $T$ exhausts all possible
linear transformations of $GF(q^n)$ over $GF(q)$, is given by
\[
\{h(x)^{e_0}(\sigma(h(x)))^{e_1}\cdots(\sigma^{(k(h)-1)}(h(x)))^{e_{k(h)-1}}|
e_j=0,1\mbox{~for~} j=0,1,\ldots,k(h)-1\}.
\]
\end{Corollary}

Now, we consider the more general situations. At first, some definitions
are needed. We define a equivalence relation
$\stackrel{\sigma}{\sim}$ on $GF(q^n)[x]$:
$f(x)\stackrel{\sigma}{\sim} g(x)$ if and only if there exists
positive integer $j$ such that $\sigma^{(j)}(f(x))=g(x)$. The
equivalence classes induced by this equivalence relation
$\stackrel{\sigma}{\sim}$ are called $\sigma-equivalence$ classes.
These definitions are introduced in \cite[Section 3]{gf}.
\begin{Theorem}
\label{Theorem2} Let $S$ be a linear recurring sequence over
$GF(q^n)$ with minimal polynomial $h(x)$ which is a product of some
distinct monic irreducible polynomials in one $\sigma-equivalence$
class. Assume that $h=h_1\sigma^{(i_1)}(h_1)\cdots
\sigma^{(i_{w-1})}(h_1)$ where $h_1(x)$ is a monic irreducible
polynomial with $\mbox{\rm{deg}}(h_1)=l$ and $i_1,i_2\ldots
i_{w-1}$ are distinct positive integers less than $k(h_1)$.
Assume by Lemma \ref{Lemma5} that $S=\sum_{j=0}^{w-1}S_j$ where
$S_j$ is the linear recurring sequence with the minimal polynomial
$\sigma^{(i_j)}(h_1)$ (here $i_0=0$). Let
$S_j(x)=g_{j}(x)/\sigma^{(i_{j})}(h_1(x))$ where
$g_{j}(x)=b_{j,0}+b_{j,1}x+\cdots+b_{j,l_j}x^{l_j}$ and $l_j<l$. Let
$T$ be a linear transformation of $GF(q^n)$ over $GF(q)$. Then, the
minimal polynomial of $T(S)$ is given by
\[h_1(x)^{e_0}(\sigma(h_1(x))^{e_1})\cdots(\sigma^{(k(h_1)-1)}(h_1(x)))^{e_{k(h_1)-1}}\]
where $e_u$, $0\leq u< k(h_1)$, is given by
\[
e_u=\mu(\sum_{j=0}^{w-1}T_{k(h_1),[u-i_j]}(b_{j,0}),\sum_{j=0}^{w-1}T_{k(h_1),[u-i_j]}(b_{j,1}),\ldots,\sum_{j=0}^{w-1}T_{k(h_1),[u-i_j]}(b_{j,l-1}))
\]
where $b_{j,k}=0$ for $l_j<k\leq l-1$, and $[i]=i$ if $0\leq
i<k(h_1)$; $[i]=k(h_1)+i$ if $-k(h_1)< i<0$, and $T_{k,j}(x)$ is
defined by (\ref{T}) and (\ref{T_{k,j}}).
\end{Theorem}
\begin{proof}
By (\ref{T}) and (\ref{T_{k,j}}) and noting that $b_{j,k}=0$ for $l_j<k\leq l-1$, we have
\begin{eqnarray}
       & &T(S)(x)\nonumber\\
       &=&\sum_{j=0}^{w-1}T(S_j(x))\nonumber\\
       &=&\sum_{j=0}^{w-1}c_0S_j(x)+c_1\sigma(S_j(x))+\cdots+c_{n-1}\sigma^{(n-1)}(S_j(x))\nonumber\\
       &=&\sum_{j=0}^{w-1}c_0\frac{g_j}{\sigma^{(i_j)}(h_1)}+c_1\frac{\sigma(g_j)}{\sigma^{(i_j+1)}(h_1)}+\cdots+c_{n-1}\frac{\sigma^{(n-1)}(g_j)}{\sigma^{(i_j+n-1)}(h_1)}\nonumber\\
       &=&\sum_{j=0}^{w-1}\sum_{u=0}^{k(h_1)-1}\frac{T_{k(h_1),u}(b_{j,0})+T_{k(h_1),u}(b_{j,1})x+\cdots+T_{k(h_1),u}(b_{j,l-1})x^{l-1}}{\sigma^{(i_j+u)}(h_1)}\nonumber\\
       &=&\sum_{j=0}^{w-1}\sum_{u=0}^{k(h_1)-1}\frac{T_{k(h_1),[u-i_j]}(b_{j,0})+T_{k(h_1),[u-i_j]}(b_{j,1})x+\cdots+T_{k(h_1),[u-i_j]}(b_{j,l-1})x^{l-1}}{\sigma^{(u)}(h_1)}\nonumber\\
       &=&\sum_{u=0}^{k(h_1)-1}\frac{\sum_{j=0}^{w-1}T_{k(h_1),[u-i_j]}(b_{j,0})+\sum_{j=0}^{w-1}T_{k(h_1),[u-i_j]}(b_{j,1})x+\cdots+\sum_{j=0}^{w-1}T_{k(h_1),[u-i_j]}(b_{j,l-1})x^{l-1}}{\sigma^{(u)}(h_1)}.\nonumber
\end{eqnarray}
Using the same argument in the proof of Theorem \ref{Theorem1}, we
have the minimal polynomial of $T(S)$ is:
\[h_1(x)^{e_0}(\sigma(h_1(x)))^{e_1}\cdots(\sigma^{(k(h_1)-1)}(h_1(x)))^{e_{k(h_1)-1}}\]
where for $0\leq u< k(h_1)$
\[
e_u=\mu(\sum_{j=0}^{w-1}T_{k(h_1),[u-i_j]}(b_{j,0}),\sum_{j=0}^{w-1}T_{k(h_1),[u-i_j]}(b_{j,1}),\ldots,\sum_{j=0}^{w-1}T_{k(h_1),[u-i_j]}(b_{j,l-1}))
\]
which completes the proof.
\end{proof}
\begin{Remark}
Note that $S_0, S_1,\ldots, S_{w-1}$ in the above theorem could be
obtained by Algorithm \ref{Algorithm1}.
\end{Remark}
\begin{Theorem}
\label{Theorem3} Let $S$ be a linear recurring sequence over
$GF(q^n)$ with minimal polynomial $h(x)$ which has no multiple
roots. Assume that the canonical factorization of $h(x)$ in
$GF(q^n)[x]$ is given by
\[
h(x)=\prod_{j=1}^{l}P_{j0}P_{j1}\cdots P_{ji_j}
\]
where $P_{ji}$ are distinct monic irreducible polynomials,
$P_{j0},P_{j1},\ldots, P_{ji_j}$ are in the same
$\sigma-equivalence$ class and $P_{tu}$, $P_{wv}$ are in different
$\sigma-equivalence$ classes when $t\neq w$. Assume by Lemma
\ref{Lemma5} that $S$ is written as
\[S=S_1+S_2+\cdots+S_l\]
where $S_1,S_2,\ldots,S_l$ are linear recurring sequences over
$GF(q^n)$ with minimal polynomial $P_{j0}P_{j1}\cdots P_{ji_j}$,
$j=1,2,\ldots,l$, respectively. Let $T$ be a linear transformation
of $GF(q^n)$ over $GF(q)$. Let $h_1(x),h_2(x),\ldots, h_l(x)$ be the
minimal polynomials of $T(S_1),T(S_2),\ldots,T(S_l)$ respectively.
Then, the minimal polynomial of $T(S)$ is given by the product of
$h_1(x),h_2(x),\ldots, h_l(x)$.
\end{Theorem}
\begin{proof}By Theorem \ref{Theorem2}, the minimal polynomial $h_j(x)$ of $T(S_j)$
must be of the form
\[
P_{j0}(x)^{e_0}(\sigma(P_{j0}(x)))^{e_1}\cdots(\sigma^{(k(P_{j0})-1)}(P_{j0}(x)))^{e_{k(P_{j0})-1}}
\]
where $e_u=0$ or $1$. Meanwhile, since there exist no positive
integer $k$ such that $P_{tu}=\sigma^{(k)}(P_{wv})$ when $t\neq w$,
we have that $h_1(x),h_2(x),\ldots, h_l(x)$ are pairwise relatively
prime. Note that $T(S)=T(S_1)+T(S_2)+\cdots+T(S_l)$. By Lemma
\ref{Lemma4}, we have the minimal polynomial of $T(S)$ is the
product of $h_1(x),h_2(x),\ldots h_l(x)$.
\end{proof}
Note that $h_j(x)$ could be obtained by Theorem \ref{Theorem2}. Now,
we are able to give our main algorithm in this paper:
\begin{Algorithm}
\label{Algorithm2} Let $S$ be a linear recurring sequence over
$GF(q^n)$ with minimal polynomial $h(x)$ which has no multiple roots.
Let $T(x)=c_0x+c_1x^q+\cdots+c_{n-1}x^{q^{n-1}}$ be a linear
transformation of $GF(q^n)$ over $GF(q)$. Assume that the canonical
factorization of $h(x)$ over $GF(q^n)$ is given by
\[h(x)=\prod_{j=1}^{m}P_j(x)
\]
Then, the procedure to find the minimal polynomial of $T(S)$ is as
follows:

Step 1. Classify $\{P_j\}$ according to the $\sigma-equivalence$
relation. Then, we get
\[
h(x)=\prod_{j=1}^{l}P_{j0}P_{j1}\cdots P_{ji_j}
\]
where $P_{ji}$ are distinct monic irreducible polynomials,
$P_{j0},P_{j1},\ldots, P_{ji_j}$ are in the same
$\sigma-equivalence$ class and $P_{tu}$, $P_{wv}$ are in different
$\sigma-equivalence$ classes when $t\neq w$.

Step 2. Use Algorithm \ref{Algorithm1} to get the decomposition of
$S$ such that $S=\sum_{j=1}^{l}S_j$ and the minimal polynomial of
$S_j$ is $P_{j0}P_{j1}\cdots P_{ji_j}$ for $1\leq j\leq l$.

Step 3. Use Theorem \ref{Theorem2} to calculate the minimal polynomial
$h_j(x)$ of $T(S_j)$ for $1\leq j\leq l$.

Step 4. By Theorem \ref{Theorem3}, the minimal polynomial of $T(S)$
is $\prod_{j=1}^{l}h_j(x)$.
\end{Algorithm}
At the end of this section, we give an example to show the procedure
of the algorithm.
\begin{Example}
\label{Example2} Let $GF(2)\subseteq GF(2^2)$ and $\alpha$ be a root
of $x^2+x+1$ in $GF(2^2)$. Let $S$ be a linear recurring sequence
with least period $15$ given by
\[
S=(1,\alpha,\alpha,0,\alpha,1,\alpha^2,\alpha^2,\alpha,\alpha+1,1,0,0,\alpha+1,0,1,\alpha,\alpha,0,\alpha,1\ldots).
\] The minimal polynomial of $S$ is $x^3+\alpha
x+\alpha^2$.  Let $T$ be a linear transformation of $GF(2^2)$ over
$GF(2)$ such that $T(x)=x+x^2$. Let us find the minimal polynomial
of
\[
T(S)=(0,1,1,0,1,0,1,1,1,1,0,0,0,1,0,0,1,1,0,1,0,\ldots).
\]
\end{Example}
The procedure is as follows:

 The canonical factorization of
$x^3+\alpha x+\alpha^2$ is
\begin{eqnarray}\label{pr2}
x^3+\alpha x+\alpha^2=(x+1)(x^2+x+\alpha^2).
\end{eqnarray}

Step 1. The canonical factorization (\ref{pr2}) has already been of
the desired form.

Step 2. Using the Division Algorithm for polynomials, we could find
$u_1(x)=\alpha x$ and $u_2(x)=\alpha$ such that $(\alpha
x)(x+1)+\alpha(x^2+x+\alpha^2)=1$. Recall that the mapping $L$ is defined in
the paragraph before Algorithm \ref{Algorithm1}. By Algorithm \ref{Algorithm1},
\[S_1=\alpha(L^2+L+\alpha^2)(S)=(1,1,1,1,1,1,1,1,1,1,1,1,1,1,1,1,1,1,1,1,1,1,\ldots)\] and the minimal polynomial of $S_1$ is $x+1$. Similarly,
\[S_2=(\alpha
L)(L+1)(S)=(0,\alpha^2,\alpha^2,1,\alpha^2,0,\alpha,\alpha,1+\alpha,\alpha,0,1,1,\alpha,1,0,\alpha^2,\alpha^2,\ldots)
\]
and the minimal polynomial of $S_2$ is $(x^2+x+\alpha^2)$.

Step 3. Note that
\[S_1(x)=\frac{1}{x+1},\quad S_2(x)=\frac{\alpha^2}{x^2+
x+\alpha^2}\] and $k(x+1)=1,k(x^2+x+\alpha^2)=2$. Then, by Theorem
\ref{Theorem2}, the minimal polynomial of $T(S_1)$ is given by
\[
h_1(x)=(1+x)^{\mu(T_{1,0}(1))}=1
\] and the minimal polynomial of $T(S_2)$ is given by
\[
h_2(x)=(x^2+x+\alpha^2)^{\mu(T_{2,0}(\alpha^2))}\sigma(x^2+x+\alpha^2)^{\mu(T_{2,1}(\alpha^2))}=(x^2+x+\alpha^2)(x^2+x+\alpha).
\]

Step 4. The minimal polynomial of $T(S)$ is the product of
$h_1(x),h_2(x)$, i.e., $(x^2+x+\alpha^2)(x^2+x+\alpha)$.

\section{On m-sequence}
\label{oms} In this section, we consider the special case in which
the minimal polynomial of $S$ is assumed to be primitive. The
following lemma is needed in this section:
\begin{Lemma}
\label{Lemma7}{\rm \cite[Lemma 3]{gx}}\quad Let $f(x)$ be a
primitive polynomial over $GF(q^n)$ with degree $m$. Let $\alpha$ be
a root of $f(x)$ in the splitting field $GF(q^{mn})$. Then, the
minimal polynomial $g(x)$ of $\alpha$ over $GF(q)$ is given by
\[g(x)=f(x)\sigma(f(x))\sigma^{(2)}(f(x))\cdots\sigma^{(n-1)}(f(x))\]
where $\sigma^{(i)}(f(x))$ is the minimal polynomial of
$\alpha^{q^i}$ over $GF(q^n)$ for $1\leq i\leq n-1$.
\end{Lemma}
Denote $Root(p(x))$ the set of the roots of $p(x)$. By the proof of
\cite[Lemma 3]{gx}, for $0\leq i\neq j\leq n-1$,
$Root(\sigma^{(i)}(f(x)))$ and $Root(\sigma^{(j)}(f(x)))$ have no
intersection and
$\cup_{i=0}^{n-1}Root(\sigma^{(i)}(f(x)))=Root(g(x))$. Therefore, we
have the following lemma:
\begin{Lemma}
\label{Lemma8} Let $f(x)$ be a primitive polynomial over $GF(q^n)$.
Then, $k(f)=n$ and
\[g(x)=f(x)\sigma(f(x))\sigma^{(2)}(f(x))\cdots\sigma^{(n-1)}(f(x))\]
is primitive over $GF(q)$.
\end{Lemma}
\begin{Theorem}
\label{Theorem4} Let $S$ be an m-sequence over $GF(q^n)$ with
primitive minimal polynomial $h(x)$. Let $T$ be a linear
transformation of $GF(q^n)$ over $GF(q)$ with
\[
T(x)=c_0x+c_1x^q+\cdots+c_{n-1}x^{q^{n-1}}\mbox{,~~~}c_i\in GF(q^n).
\] Then, the minimal polynomial of $T(S)$ is given by
\[
h(x)^{e_0}(\sigma(h(x)))^{e_1}\cdots(\sigma^{(n-1)}(h(x)))^{e_{n-1}}
\]where $e_j=0$ if $c_j=0$ and $e_j=1$ if $c_j\neq 0$.
\end{Theorem}
\begin{proof} Assume that $S(x)=g(x)/h(x)$ where
$g(x)=b_0+b_1x+\cdots+b_lx^l\neq 0$ and $l<\mbox{deg}(h(x))$. Then,
we have $(b_0,b_1,\ldots,b_l)\neq (0,0,\ldots,0)$. It is known from Lemma
\ref{Lemma8} that $k(h)=n$. Then, by $(\ref{T})$ and $(\ref{T_{k,j}})$, we have
$T_{n,j}(x)=c_j x^{q^j}$ for $0\leq j<n$. By
Theorem \ref{Theorem1}, the minimal polynomial of $T(S)$ is:
\[h(x)^{e_0}(\sigma(h(x)))^{e_1}\cdots(\sigma^{(n-1)}(h(x)))^{e_{n-1}}\]
where
$e_j=\mu(T_{n,j}(b_0),T_{n,j}(b_1),\ldots,T_{n,j}(b_l))=\mu(c_jb_0^{q^j},c_jb_1^{q^j},\ldots,c_jb_l^{q^j})$.
Note that $(b_0^{q^j},b_1^{q^j},\ldots,b_l^{q^j})\neq (0,0,\ldots,0)$ since $(b_0,b_1,\ldots,b_l)\neq (0,0,\ldots,0)$. Thus,
by the definition of $\mu$, we have
$e_j=\mu(c_jb_0^{q^j},c_jb_1^{q^j},\ldots,c_jb_l^{q^j})=\mu(c_j)$.
Therefore, the minimal polynomial of $T(S)$ is:
\[
h(x)^{e_0}(\sigma(h(x)))^{e_1}\cdots(\sigma^{(n-1)}(h(x)))^{e_{n-1}}
\]where $e_j=0$ if $c_j=0$ and $e_j=1$ if $c_j\neq 0$. This completes the proof.
\end{proof}
\begin{Remark}
\label{Remark2} It is known from Theorem \ref{Theorem4} that the main result
\cite[Theorem 1]{yg} is correct when the linear recurring sequence
$S$ over $GF(q^n)$ is an m-sequence.
\end{Remark}
\begin{Corollary}
\label{Corollary3} Suppose that $S,h(x),T(x)$ are defined as the
same as in Theorem \ref{Theorem4}. Assume that $c_j\neq 0$ for all
$0\leq j<n$. Then, the minimal polynomial of $T(S)$ is
\[
h(x)\sigma(h(x))\cdots\sigma^{(n-1)}(h(x))\] and the linear
complexity of such $T(S)$ is the maximum among all possible linear
transformations $T$ of $GF(q^n)$ over $GF(q)$.
\end{Corollary}
At the end of this section, we consider a special linear
transformation of $GF(q^n)$ over $GF(q)$, trace function.
\begin{Corollary}
\label{Corollary4}Let $Tr$ be the trace function from $GF(q^n)$ to
$GF(q)$, i.e,
\[
Tr(x)=x+x^q+x^{q^2}+\cdots+x^{q^{n-1}}
\] and let $h(x)$ be a primitive polynomial over $GF(q^n)$. Then,
$Tr$ is a bijective mapping from the set of all m-sequences over
$GF(q^n)$ with primitive minimal polynomial $h(x)$ to the set of all
m-sequences over $GF(q)$ with primitive minimal polynomial
$g(x)=h(x)\sigma(h(x))\cdots\sigma^{(n-1)}(h(x))$ over $GF(q)$.
\end{Corollary}
\begin{proof} Let $S$ be an m-sequence over $GF(q^n)$ with minimal
polynomial $h(x)$. By Lemma \ref{Lemma8} and Corollary
\ref{Corollary3} , we have $Tr(S)$ is an m-sequence over $GF(q)$
with primitive minimal polynomial
$g(x)=h(x)\sigma(h(x))\cdots\sigma^{(n-1)}(h(x))$ over $GF(q)$. For
any two different m-sequences $S_1,S_2$ over $GF(q^n)$ with minimal
polynomial $h(x)$, we have $S_1-S_2$ is also an m-sequence over
$GF(q^n)$ with minimal polynomial $h(x)$. Thus, $Tr(S_1-S_2)$ is an
m-sequence over $GF(q)$ with minimal polynomial $g(x)$. Then,
$Tr(S_1-S_2)\neq 0$. So, $Tr(S_1)\neq Tr(S_2)$ which implies that
$Tr$ is injective. Since
$h(x),\sigma(h(x)),\ldots,\sigma^{(n-1)}(h(x))$ have the same
degree, then the number of m-sequences over $GF(q^n)$ with minimal
polynomial $h(x)$ is equal to that of m-sequences over $GF(q)$ with
minimal polynomial $g(x)$. Therefore, $Tr$ is bijective. This
completes the proof.
\end{proof}
\section{Upper bound on linear complexity of $T(S)$}
\label{ublc} In this section, we give the definition of linear
complexity over $GF(q)$ of linear recurring sequence $S$ over
$GF(q^n)$ and then show it is an upper bound on the linear
complexity of $T(S)$ when $T$ exhausts all possible linear
transformations of $GF(q^n)$ over $GF(q)$. The notion of linear
complexity over $GF(q)$ of linear recurring sequences over $GF(q^n)$
was introduced by Ding, Xiao and Shan in \cite{dxs}, and discussed
by some authors, for example, see \cite{gf}, \cite{cm}-\cite{nv}.

Let $S=(s_1,s_2,\ldots,s_m,\ldots)$ be a linear recurring sequence
over $GF(q^n)$. The polynomial $f(x)=a_0+a_1x+\cdots+a_mx^{m}$ over
$GF(q^n)$ is called a characteristic polynomial over $GF(q^n)$
of $S$ if
\[
a_0s_i+a_1s_{1+i}+\cdots+a_ms_{m+i}=0 \quad \mbox{for }i\geq 1.
\]
If the characteristic polynomial $f(x)$ is a polynomial over $GF(q)$,
that is, all $a_i \in GF(q)$, we call $f(x)$ a characteristic
polynomial over $GF(q)$ of $S$. The monic characteristic polynomial
over $GF(q^n)$ (resp. $GF(q)$) of $S$ with least degree is called
the minimal polynomial over $GF(q^n)$ (resp. $GF(q)$) of $S$. The
degree of the minimal polynomial over $GF(q^n)$ (resp. $GF(q)$) of
$S$ is called the linear complexity over $GF(q^n)$ (resp. $GF(q)$)
of $S$. We will use the following lemma in this section.
\begin{Lemma}
\label{Lemma9}{\rm \cite[Theorem 5]{gf}}\quad Let $S$ be a linear
recurring sequence over $GF(q^n)$ with minimal polynomial $h(x)\in
GF(q^n)[x]$. Assume that the canonical factorization of $h(x)$ in
$GF(q^n)[x]$ is given by
\[
h(x)=\prod_{j=1}^{l}P_{j0}^{e_{j0}}P_{j1}^{e_{j1}}\cdots
P_{ji_j}^{e_{ji_j}}
\]
where $\{P_{uv}\}$ are distinct monic irreducible polynomials in
$GF(q^n)[x]$, $P_{j0},P_{j1},\ldots, P_{ji_j}$ are in the same
$\sigma$-equivalence class and $P_{uv}$, $P_{tw}$ are in the
different $\sigma$-equivalence classes when $u\neq t$. Then the
minimal polynomial over $GF(q)$ of $S$ is given by
\[
H(x)=\prod_{j=1}^{l}R(P_{j0})^{e_j}
\]
where $e_j=\max\{e_{j0},e_{j1},\ldots,e_{ji_j}\}$ and
$R(P_{j0})=P_{j0}\sigma(P_{j0})\cdots
\sigma^{(k(P_{j0})-1)}(P_{j0})$ for $1\leq j\leq l$.
\end{Lemma}
Denote $L_{q^n}(S)$ the linear complexity over $GF(q^n)$ of $S$ and
$L_{q}(S)$ the linear complexity over $GF(q)$ of $S$. Note that
$L_{q^n}(S)$ is the linear complexity of $S$, which is discussed in
the previous sections. Then, we have the following theorem:
\begin{Theorem}
\label{Theorem6} Let $S$ be a linear recurring sequence over
$GF(q^n)$ with minimal polynomial $h(x)$ over $GF(q^n)$ which has no
multiple roots. Then, for any linear transformation $T$ of $GF(q^n)$
over $GF(q)$, we have
\[
L_{q^n}(T(S))\leq L_{q}(S).
\]
\end{Theorem}
\begin{proof}
Recall that $\{P_{ji}\}$, $\{S_j\}$ and $\{h_j\}$ are defined in
Theorem \ref{Theorem3}. Note that the minimal polynomial over
$GF(q^n)$ of $T(S)$ is the product of $\{h_j\}$. By Theorem
\ref{Theorem2}, we have $h_j|R(P_{j0})$ where $R(P_{j0})$ is defined
in Lemma \ref{Lemma9}. Then, by Theorem \ref{Theorem3} and Lemma
\ref{Lemma9}, the minimal polynomial over $GF(q)$ of $S$ is a
multiple of the minimal polynomial over $GF(q^n)$ of $T(S)$.
Therefore, $L_{q^n}(T(S))\leq L_{q}(S).$
\end{proof}
\begin{Remark}
\label{Remark3} Theorem \ref{Theorem6} gives an upper bound on the
linear complexity of $T(S)$. By Corollary \ref{Corollary1}, we know
that this bound is tight if the minimal polynomial $h(x)$ over
$GF(q^n)$ of $S$ is irreducible over $GF(q^n)$.
\end{Remark}

\end{document}